\documentclass[copyright]{eptcs}
\providecommand{\event}{J.~Dassow, G.~Pighizzini, B.~Truthe (Eds.): 11th International Workshop\\
on Descriptional Complexity of Formal Systems (DCFS 2009)} 
\providecommand{\volume}{3}
\providecommand{\anno}{2009}
\providecommand{\firstpage}{179} 
\providecommand{\eid}{17}
\usepackage{breakurl}        

\input{dcfs.tex}

\newcommand{\bean}{\begin{eqnarray*}}
\newcommand{\eean}{\end{eqnarray*}}
\newcommand{\bp}{\begin{proof}}
\newcommand{\ep}{\end{proof}}
\newcommand{\ra}{\rightarrow}
\newcommand{\lf}{\langle}
\newcommand{\rf}{\rangle}
\newcommand{\lc}{\prec\!\!}
\newcommand{\rc}{\!\!\succ}
\newcommand{\Lra}{\Longrightarrow}

\newcommand{\lra}{\longrightarrow}

\def\eps{\varepsilon}
\newcommand{\nat}{\mathbb{N}}
\newcommand{\ANEPFC}{\mathit{ANEPFC}}
\newcommand{\dollar}{\$}

\newcommand{\s}{\sigma}

\newcommand{\y}{\,\wedge\,}

\newtheorem{cor}{Corollary}

\newcommand{\bdisplay}{\begin{description}\footnotesize\item[]}
\newcommand{\edisplay}{\end{description}}

\newcommand{\bquot}[1]{\begin{quotation}\small\noindent
  \textbf{#1}\hspace{\labelsep}\ignorespaces}
\newcommand{\equot}{\unskip\end{quotation}}

\begin{document}%
\title{Small Universal Accepting Networks of Evolutionary Processors with Filtered Connections\,%
\thanks{Remco Loos' work was supported by Research Grant ES-2006-0146 of the
Spanish Ministry of Science and Innovation. Victor Mitrana
acknowledges support from the Alexander von Humboldt Foundation and
the Academy of Finland, project 132727.}
}
\def\titlerunning{Small Universal ANEPFCs}
\def\authorrunning{R.~Loos, F.~Manea, V.~Mitrana}
\author{Remco Loos
\institute{EMBL -- European Bioinformatics Institute\\
Wellcome Trust Genome Campus --
Hinxton -- Cambridge -- CB10 1SD -- UK}
\email{remco.loos@ebi.ac.uk}
\and
Florin Manea
\institute{Faculty of Mathematics and Computer Science -- University of Bucharest\\
Academiei 14 -- Bucharest -- 010014 -- Romania}
\email{flmanea@gmail.com}
\and
Victor Mitrana
\institute{Department of Information Systems and Computation --
Technical University of Valencia\\
Camino de Vera s/n. -- 46022 Valencia -- Spain}
\email{mitrana@fmi.unibuc.ro}
}

\maketitle
\begin{abstract}
  In this paper, we present some results regarding the size complexity of Accepting Networks of Evolutionary Processors with Filtered Connections (ANEPFCs). We show that there are universal ANEPFCs of size $10$, by devising a method for simulating 2-Tag Systems. This result significantly improves the known upper bound for the size of universal ANEPFCs which is $18$.
  We also propose a new, computationally and descriptionally efficient simulation of nondeterministic Turing machines by ANEPFCs. More precisely, we 
  describe (informally, due to space limitations) how ANEPFCs with $16$ nodes can simulate in $O(f(n))$ time any nondeterministic Turing machine of time complexity $f(n)$. Thus the known upper bound for the number of nodes in a network simulating an arbitrary Turing machine is decreased from $26$ to $16$.
\end{abstract}

%
%
\section{Introduction}

The basic structure of an accepting network of evolutionary processors (ANEP for short) is
widely met in distributed, parallel and evolutionary computing: a virtual undirected graph whose
nodes are processors handling some data. All node processors act simultaneously on the local data
in accordance with some predefined rules, and then local data becomes a mobile agent which can
navigate in the network following a given protocol.
All the nodes send simultaneously their data and the receiving nodes handle also simultaneously all the arriving messages. Only the data able to pass a filtering process can be communicated to the other processors. This filtering process may require the data to satisfy some conditions imposed by the sending processor, by the receiving processor or by both of them.

In a series of papers starting with \cite{gnep} (for the generating variants)
and \cite{dna10} (for the accepting variants) this general structure is particularized in a
bio-inspired way: each node may be viewed as a cell having genetic information encoded in DNA sequences which may evolve by very simple evolutionary events, that is point mutations. Each node is specialized in just one of these evolutionary operations. Furthermore, the data in each node is organized in the form of multisets of words, each word appearing in an arbitrarily large number of copies, and all the copies are processed in a massive parallel manner, such that all the possible events that can take place do actually take place.
Furthermore, the filtering process is based on simple ``random-context" conditions, namely the presence/absence of some symbols. Clearly, the biological hints presented above are intended to explain in an informal way how some biological phenomena are {\it sources of inspiration} for the mathematical computing model.

In \cite{dna10} one presents a characterization of the complexity class {\bf NP} based on ANEPs. The work \cite{mscs} discusses how ANEPs can be considered as problem solvers. In \cite{ipl}, one shows that every recursively enumerable language can be accepted by an ANEP with $24$ nodes. Moreover, this construction proves that for every {\bf NP}-language there exists an ANEP of size $24$ deciding that language in polynomial time. While the number of nodes of this ANEP does not depend on the language, the other parameters of the network (rules, symbols, filters) depend on it. This result may be also interpreted as a method for solving every {\bf NP}-problem in polynomial time by ANEPs of constant size.
All the aforementioned results were obtained via simulations of Turing machines by ANEPs.

It is expected that having filters associated with each node, as in the case of ANEPs, allows a strong control of the computation. Indeed, every node
has an associated input and output filter; two nodes can exchange data if it passes the output filter of the sender {\it and} the input filter of the receiver. Moreover, if some data is sent out by a node and not able to enter any other node, then it is lost. In \cite{jucsTOM} the ANEP model considered in \cite{dna10} was simplified by moving the filters from the nodes to the edges. Each edge was viewed as a two-way channel such that the input and output filters, respectively, of the two nodes connected by the edge coincide; now two nodes can exchange data if it passes the filters of the edge existing between the two nodes. Clearly, the possibility of controlling the computation in such networks seemed to be diminished. For instance, there was no possibility to loose data during the communication steps. However, in \cite{jucsTOM}
(a simplified proof can be found in \cite{{ijfcsDM}}) one proves that these newly introduced
devices, called accepting networks of evolutionary processors with filtered connections (ANEPFCs for short), were still computationally complete. Furthermore,
in \cite{ijfcsDM} it is constructed an universal ANEPFC with $18$ nodes, and it is shown that every recursively enumerable language $L$ is recognized by an ANEPFC with $27$ processors, having the property that only $7$ of its nodes depend on the language $L$, while the others remain unchanged regardless the accepted language; moreover, any nondeterministic Turing machine can be simulated efficiently (with respect to time and space complexity) by an ANEPFC with $26$ processors. All these results were also based on
simulations of Turing machines by ANEPFCs.

Here we aim to improve the results reported in \cite{ijfcsDM}. More precisely, we first look for a universal ANEPFC with a smaller number of nodes. To this end, we propose a simulation of 2-tag systems introduced in \cite{minsky}. Second, we are interested in finding a way to design ANEPFCs with less than $26$ nodes that simulate computationally efficient nondeterministic Turing machines.
Note that the universal ANEPFC obtained
from the simulation of a tag system does not solve the second problem as a $2$-tag system can efficiently simulate any deterministic Turing machine but not nondeterministic ones.
Based on a similar idea to that used in the simulation of tag systems, we propose a simulation of nondeterministic Turing machines with ANEPFCs of size $16$ which maintain the working time of the Turing machine. That is, every language accepted by a nondeterministic Turing machine in time  $f(n)$ can be accepted by an ANEPFC of size $16$ in time $O(f(n))$. Consequently, the class ${\bf NP}$ equals the class of languages accepted in polynomial time by ANEPFCs of size $16$. This result considerably improves the known bound of $26$ reported in \cite{ijfcsDM}.

\section{Basic definitions}

We start by summarizing the notions used throughout the paper; for all unexplained notions
the reader is referred to \cite{handbook}. An
{\it alphabet} is a finite and nonempty set of symbols. The
cardinality of a finite set $A$ is written $card(A)$. Any sequence
of symbols from an alphabet $V$ is called {\it word (string)} over
$V$. The set of all words over $V$ is denoted by $V^*$ and the
empty word is denoted by $\eps$. The length of a word $x$ is
denoted by $|x|$ while $alph(x)$ denotes the minimal alphabet $W$ such
that $x\in W^*$. For a word $x\in W^*$, $x^r$ denotes the reversal of the word.

We consider here the following definition of $2$-tag systems that appears in~\cite{rogo}.
It is slightly different but equivalent to those from \cite{post,minsky}. A $2$-tag system 
\hbox{$T=(V,\phi)$} consists of a finite alphabet of symbols~$V$, containing a special {\it halting symbol} $H$ and a finite set of rules $\phi:V\setminus\{H\} \ra V^+$ such that $|\phi(x)|\geq 2$ or $\phi(x)=H$.
Furthermore, $\phi(x)=H$ for just one $x\in V\setminus\{H\}$.
A halting word for the system $T$ is a word that contains the halting symbol $H$ or whose length is less than $2$; the transformation $t_T$ (called the tag operation) is defined on the set of non-halting words as follows: if $x$ is the leftmost symbol of a non-halting word $w$, then $t_T(w)$ is the result of deleting the leftmost $2$ symbols of $w$ and then appending the word $\phi(x)$ at the right end of the obtained word.
A computation by a $2$-tag system as above is a finite sequence of words produced by iterating the transformation $t$, starting with an initially given non-halting word $w$ and halting when a halting word is produced. Note that a computation is not considered to exist unless a halting word is produced in finitely-many iterations. We recall that such restricted $2$-tag systems are universal \cite{rogo}.

A nondeterministic Turing machine is a construct $M=(Q,$ $V,$ $U,$
$\delta ,$ $q_0,$ $B,$ $F)$, where $Q$ is a finite set of states,
$V$ is the input alphabet, $U$ is the tape alphabet, $V \subset
U$, $q_0$ is the initial state, $B \in U \setminus V$ is the
``blank'' symbol, $F \subseteq Q$ is the set of final states, and
$\delta$ is the transition mapping,\linebreak 
$ \delta : (Q\setminus F)\!\times\! U\! \rightarrow\! 2^{Q\times (U\setminus
\{B\})\times\{R,L\}}$. In this paper, we assume without
loss of generality that any Turing machine we consider has a
semi-infinite tape (bounded to the left) and makes no stationary
moves; the computation of such a machine is described in \cite{handbook,hartmanis2,Papa}. An input word is accepted if and only if after a finite number of moves the
Turing machine enters a final state. The language accepted by the
Turing machine is a set of all accepted words. We say a Turing
machine \emph{decides} a language $L$ if it accepts $L$ and
moreover halts on every input.
The reader is referred to \cite{hartmanis2,Papa} for the classical
time and space complexity classes defined for Turing machines.

We say that a rule $a\ra b$, with $a,b\in V\cup\{\eps\}, a\neq b,$  is a {\it substitution rule} if both $a$ and $b$ are not $\eps$; it is a {\it deletion rule} if $a\ne\eps$ and $b=\eps$; it is an {\it insertion rule} if $a=\eps$ and $b\ne\eps$. The set of all substitution, deletion, and insertion rules over an alphabet $V$ are denoted by $Sub_V$, $Del_V$, and $Ins_V$, respectively.

Given a rule as above $\s$ and a word $w\in V^*$, we define the following
\emph{actions} of $\s$ on $w$:
\begin{itemize}
\item If $\s\equiv a\ra b\in Sub_V$, then
$$\s^*(w)=\s^r(w)=\s^l(w)=\left\{
\begin{array}{ll}
\{ubv:\ \exists u,v\in V^*\ (w=uav)\},\\
\{w\},\mbox{ otherwise}.
\end{array}\right.$$

\item If $\s\equiv a\to \eps\in Del_V$, then
$
\s^*(w)=\left\{
\begin{array}{ll}
\{uv:\ \exists u,v\in V^*\ (w=uav)\},\\
\{w\},\mbox{ otherwise},
\end{array}\right.$
$$
\begin{array}{llc}
\s^r(w)=\left\{
\begin{array}{ll}
\{u:\ w=ua\},\\
\{w\},\mbox{ otherwise},
\end{array}\right. & \qquad &
\s^l(w)=\left\{
\begin{array}{ll}
\{v:\ w=av\},\\
\{w\},\mbox{ otherwise}.
\end{array}\right.
\end{array}
$$

\item If $\s\equiv \eps\to a\in Ins_V$, then 
$\s^*(w)=\{uav:\ \exists u,v\in V^*\ (w=uv)\},\ \s^r(w)=\{wa\},\
\s^l(w)=\{aw\}.$
\end{itemize}


\noindent In the following $\alpha\in\{*,l,r\}$ expresses the way of applying a deletion or insertion rule to a word, namely at any position ($\alpha=*$), at the left ($\alpha=l$), or at the right ($\alpha=r$) end of the word, respectively.  For every rule $\s$, action $\alpha\in \{*,l,r\}$, and $L\subseteq V^*$, we define the \emph{$\alpha$-action of $\s$ on $L$} by $\s^\alpha(L)=\bigcup_{w\in L} \s^\alpha(w)$.  Given a finite set of rules $M$, we define the \emph{$\alpha$-action of $M$} on the word $w$ and the language $L$ by
\begin{center}$ M^{\alpha}(w)=\bigcup_{\s\in M}
\s^{\alpha}(w)\ \mbox{ and } \ M^{\alpha}(L)=\bigcup_{w\in
L}M^{\alpha}(w),$\end{center}
respectively.  In what follows, we shall refer to the rewriting operations defined above as {\it evolutionary operations} since they may be viewed as linguistic formulations of local gene mutations.

For two disjoint subsets $P$ and $F$ of an alphabet $V$ and a word $x$ over $V$, we define the predicates
\begin{align*}
 \varphi^{s}(x;P,F)&\equiv P\subseteq alph(x) \y F\cap alph(x)=\emptyset,\\
 \varphi^{w}(x;P,F)&\equiv alph(x)\cap P \ne \emptyset \y F\cap alph(x)=\emptyset.
\end{align*}

The construction of these predicates is based on {\it random-context conditions} defined by the two sets $P$ ({\it permitting contexts/symbols})
and $F$ ({\it forbidding contexts/symbols}). Informally, the former condition requires ($s$ stands for strong) that all permitting symbols are and no forbidding symbol is present in $x$,  while the latter ($w$ stands for weak) is a weaker variant such that at least one permitting symbol appears in $x$ but still no forbidding symbol is present in $x$.

For every language $L\subseteq V^*$, $P$, $F$ as above, and $\beta\in \{s,w\}$, we define:
\begin{center}$\varphi^\beta(L,P,F)=\{x\in L\mid \varphi^\beta(x;P,F)\}.$\end{center}

An \emph{accepting network of evolutionary processors with filtered connections} (abbreviated ANEPFC) is a $9$-tuple
\begin{center}
$\Gamma=(V,U,G,{\cal R},{\cal N},\alpha,\beta,x_I,x_O),$
\end{center}
where:
\begin{itemize}
\item $V$ and $U$ are the \emph{input and network alphabet}, respectively; we have $V\subseteq U$.
\item $G=(X_G,E_G)$ is an undirected graph without loops with the set of nodes
$X_G$ and the set of edges $E_G$. Each edge is given in the form of a
binary set. $G$ is called the \emph{underlying graph} of the network.
\item ${\cal R}:X_G\lra 2^{Sub_U}\cup 2^{Del_U}\cup 2^{Ins_U}$ is a mapping which associates 
with each node \emph{the set of evolutionary rules} that can be applied in that node. 
Note that each node is associated only with one type of evolutionary rules, namely for 
every $x\in X_G$ either ${\cal R}(x)\subset Sub_U$ or ${\cal R}(x)\subset Del_U$
or\linebreak ${\cal R}(x)\subset Ins_U$ holds.
\item ${\cal N}:E_G\lra 2^U\times 2^U$ is a mapping which associates with each
edge \hbox{$e\in E_G$} \emph{the permitting and forbidding filters of that edge}; 
formally, \hbox{${\cal N}(e)=(P_e,F_e)$}, with $P_e\cap F_e=\emptyset$.
\item $\alpha: X_G\lra \{*,l,r\}$; $\alpha(x)$
gives \emph{the action mode of the rules} of node $x$ on the words existing in that
node.
\item $\beta: E_G\lra \{s,w\}$ defines \emph{the filter type of an edge}.
\item $x_I, x_O \in X_G$  are \emph{the input and the output node} of $\Gamma$, respectively.
\end{itemize}
We say that $card(X_G)$ is the size of $\Gamma$. Generally, the ANEPs considered in the literature have complete underlying graphs, namely graphs without loops in which every two nodes are connected. Starting from the observation that every ANEPFC can be immediately transformed into an equivalent ANEPFC with a complete underlying graph (the edges that are to be added are associated with filters which make them useless), for the sake of simplicity, we discuss in what follows ANEPFCs whose underlying graphs have useful edges only. Note that this is not always possible for ANEPs.

A \emph{configuration} of an ANEPFC $\Gamma$ as above is a mapping $C:X_G\lra 2^{U^*}$ which associates a set of words with every node of the graph. A configuration may be understood as the sets of words which are present in any node at a given moment. Given a word $w \in V^*$, \emph{the initial configuration} of $\Gamma$ on $w$ is defined by $C_0^{(w)}(x_I)=\{w\}$ and $C_0^{(w)}(x)=\emptyset$ for all $x\in X_G\setminus\{x_I\}$.

A configuration can change either by an {\it evolutionary step} or by a {\it communication step}. When changing by an evolutionary step, each component $C(x)$ of the configuration $C$ is changed in accordance with the set of evolutionary rules ${\cal R}(x)$ associated with the node $x$ and the way of applying these rules $\alpha(x)$. Formally, we say that the configuration $C'$ is obtained in \emph{one evolutionary step} from the configuration $C$, written as $C\Lra C'$, iff 
$C'(x)=({\cal R}(x))^{\alpha(x)}(C(x))\mbox{ for all } x\in X_G.$

When changing by a communication step, each node processor $x\in X_G$ sends one copy of each word it contains to every node processor $y$ connected to $x$, provided they can pass the filter of the edge between $x$ and $y$. It keeps no copy of these words but receives all the words sent by any node processor $z$ connected with $x$ providing that they can pass the filter of the edge between $x$ and $z$.

Formally, we say that the configuration $C'$ is obtained in \emph{one communication step} from configuration~$C$, written as $C\vdash C'$, iff
\bean
C'(x)&=&(C(x)\setminus(\bigcup_{\{x,y\}\in E_G}\varphi^{\beta(\{x,y\})}(C(x),{\cal N}(\{x,y\}))))\cup (\bigcup_{\{x,y\}\in E_G} \varphi^{\beta(\{x,y\})}(C(y),{\cal N}(\{x,y\})))
\eean
for all $x\in X_G.$

Let $\Gamma$ be an ANEPFC; the computation of $\Gamma$ on the input word $z\in V^*$ is a sequence of configurations $C_0^{(z)},C_1^{(z)},C_2^{(z)},\dots$, where $C_0^{(z)}$ is the initial configuration of $\Gamma$ on $z$, $C_{2i}^{(z)}\Lra C_{2i+1}^{(z)}$ and
$C_{2i+1}^{(z)}\vdash C_{2i+2}^{(z)}$, for all $i\geq 0$. By the previous definitions, each configuration $C_i^{(z)}$ is uniquely determined by the configuration $C_{i-1}^{(z)}$, thus each computation in an ANEPFC can be seen as deterministic.

A computation {\it halts} (and it is said to be {\it finite}) if one of the following two conditions holds:\\
(i) There exists a configuration in which the set of words existing in the output node $x_O$ is non-empty. In this case, the computation
is said to be an {\it accepting computation}.\\
(ii) There exist two identical configurations obtained either in consecutive evolutionary steps or in consecutive communication steps.

The {\it language accepted} by $\Gamma$ is
$$L_a(\Gamma)=\{z\in V^*\mid \mbox{ the computation of $\Gamma$ on $z$}\mbox{ is an accepting one.}\}$$
We say that an ANEPFC $\Gamma$ decides the language $L\subseteq V^*$, and write $L(\Gamma)=L$ iff $L_a(\Gamma)=L$ and the computation of $\Gamma$ on every $z\in V^*$ halts.

In a similar way to Turing machines, we define two computational complexity measures 
using ANEPFC as the computing model. To this aim we consider an ANEPFC $\Gamma$ with the
input alphabet~$V$ that halts on every input. The {\it time complexity} of the finite 
computation $C_0^{(x)}, C_1^{(x)}, C_2^{(x)}, \dots, C_m^{(x)}$ of $\Gamma$ on $x\in V^*$ is denoted by $Time_{\Gamma}(x)$ and equals $m$.
The time complexity of $\Gamma$ is the partial function from $\nat$ to $\nat$:
$Time_{\Gamma}(n)=\mbox{max}\{Time_{\Gamma}(x)\mid x\in
V^*, |x|=n\}.$
We say that $\Gamma$ decides $L$ in time $O(f(n))$ if $Time_{\Gamma}(n)\in O(f(n))$.

For a function $f:{\nat}\lra {\nat}$ we define:
\bean
{\bf Time}_{\ANEPFC_p}(f(n))&=&\{L\mid \mbox{there exists an
ANEPFC $\Gamma$ of size $p$ deciding $L$,}\\ 
&& \quad\mbox{and $n_0$ such that 
$Time_{\Gamma}(n) \le f(n)$ for all $n \geq n_0$}\}.
\eean
Moreover, we write ${\bf PTime}_{\ANEPFC_p}=\displaystyle\bigcup_{k\ge 0} {\bf Time}_{\ANEPFC_p}(n^k)$ for all $p\ge 1$ as well as 
$${\bf PTime}_{\ANEPFC}=\displaystyle\bigcup_{p\ge 1} {\bf PTime}_{\ANEPFC_p}.$$

We recall from \cite{ijfcsDM}:

\begin{theorem}\label{result_ipl}
${\bf NP} = {\bf PTime}_{\ANEPFC_{26}}$.
\end{theorem}

\section{Decreasing the size of universal ANEPFCs}

In the following we show how a $2$-tag system can be simulated by an ANEPFC of size $10$.
\begin{theorem}\label{2tag}
For every $2$-tag system $T=(V,\phi)$ there exists a complete ANEPFC $\Gamma$
of size $10$ such that $L(\Gamma)=\{w\mid \mbox{ $T$ halts on $w$}\}$.
\end{theorem}
\begin{proof} Let $V=\{a_1,a_2,\ldots,a_n,a_{n+1}\}$ be the alphabet of the tag system $T$
with $a_{n+1}=H$ and\linebreak \hbox{$V'=V\setminus \{H\}$}.
We consider the ANEPFC 
$$\Gamma=(V',U,K_{10},{\cal R},{\mathcal N},\alpha,\beta,1, 10)$$
with the $10$ nodes labeled with the numbers from $1$ to $10$.
The working alphabet of the network is defined as follows:\\
\centerline{$U=V\cup\{\dollar,\#,a_0',a_0'',\lc a_0\rc'\}\cup\{a',a'',a^{\circ}\mid a\in V'\}\cup
\{[x],\lf x \rf ,\lc x\rc \mid x\in X\}$,}\\
where $X=\{x\in (V\cup\{a_0\})^*\mid |x|\leq \mbox{max}\{|\phi(a)|\mid a\in V'\}\}, $ and $\dollar,\# \notin V $.
The processors placed in the nodes of the network are defined as follows (we assume that the output node $10$ has an empty set of rules):

\begin{itemize}\small
\item \underline{The node $1$}:
 $M = \{ a\ra [\phi(a)], a\ra a^{\circ}\mid a\in V'\}$, $\alpha = *$.

\item \underline{The node $2$}:
 $M = \{\eps\ra a_0'\}$, $\alpha= r$.

\item \underline{The node $3$}:
 $M = \{\eps\ra \dollar\}$, $\alpha= r$.

\item \underline{The node $4$}:
 $M = \{a_{k}'\ra a_{k}''\mid 0\le k\le n\}\cup{}$\\
 \hspace*{\fill}$\{[a_kx]\ra \lf a_{k-1}x\rf, \lc a_kx \rc\ra \lf a_{k-1}x\rf,
\lc a_0a_kx\rc\ra \lf a_{k-1}x\rf \mid x\in X, 1\!\le\! k\!\le\! n+1\}$,\\
 \hspace*{\fill}$\alpha = *$.

\item \underline{The node $5$}:
 $M = \{a_{k-1}''\ra a_{k}'\mid 1\le k\le n\}\cup\{a_{k-1}''\ra a_{k}\mid 1\le k\le n+1\}
\cup \{\lf x\rf\ra \lc x\rc \mid x\in X\}$, $\alpha = *$.

\item \underline{The node $6$}:
 $M = \{\dollar\ra \eps\}$, $\alpha= r$.

\item \underline{The node $7$}:
 $M = \{\lc a_0\rc\ra\lc a_0\rc'\}$, 
 $\alpha =* $.

\item \underline{The node $8$}:
 $M = \{\lc a_0\rc'\ra \eps\}$, $\alpha= l$.

\item \underline{The node $9$}:
 $M = \{a^{\circ}\ra\eps\mid a\in V' \}$, $\alpha =l $.

\end{itemize}

The edges of the network and their filters are defined as follows:

\begin{itemize}\small
\item \underline{The edge $\{1,2\}$} has $\beta = w$ and
$P = \{a^\circ\mid a\in V\},\
F = \{a'\mid a\in V\}\cup \{a_{n+1},a_0'\}.$

\item \underline{The edge $\{2,3\}$} has $\beta = w$ and
$ P = \{[\phi(a)]\mid a\in V'\},\
 F = \{\dollar, \lc a_0\rc\}.$

\item \underline{The edge $\{3,4\}$} has $\beta = w$ and
$ P = \{\dollar\},\
F = \{\lc a_0\rc\}\cup \{\lf x \rf \mid x\in X\}\cup \{a''\mid a\in V\cup\{a_0\}\}.$

\item \underline{The edge $\{4,5\}$} has $\beta = w$ and
$ P = \{a'' \mid a\in V\cup\{a_0\}\},\ 
F = \{[x]\mid x\in X\}\cup \{\lc x\rc \mid x\in X\} .$

\item \underline{The edge $\{5,6\}$} has $\beta = w$ and
$ P = \{\lc x\rc \mid x\in X\},\
F = \{a''\mid a\in V\cup\{a_0\}\}\cup\{[x],\lf x\rf\mid x\in X\}.$

\item \underline{The edge $\{6,2\}$} has $\beta = w$ and
\begin{align*} 
P &= \{\lc a_0x\rc\mid x\in X, x\neq \lambda\},\\
F &= \{\dollar,a_0'\}\cup\{[x],\lf x\rf\mid x\in X\}\cup \{\lc x\rc\mid x\neq a_0y, y\in X\}.
\end{align*}

\item \underline{The edge $\{6,3\}$} has $\beta = w$ and
\begin{align*}
P &= \{\lc x\rc\mid x\neq a_0y, y\in X\},\\
F &= \{\dollar\}\cup\{[x],\lf x\rf\mid x\in X\}\cup \{\lc x\rc\mid x= a_0y, y\in X\}.
\end{align*}

\item \underline{The edge $\{6,7\}$} has $\beta = w$ and
$ P = \{\lc a_0\rc \},\
F = \{\dollar\}\cup\{\lc x\rc\mid x\in X, x\neq a_0\}\cup\{[x],\lf x\rf\mid x\in X\}.$

\item \underline{The edge $\{7,8\}$} has $\beta = w$ and
$ P = \{\lc a_0\rc'\},\
F =\{\lc a_0\rc \}.$

\item \underline{The edge $\{7,3\}$} has $\beta = w$ and
$ P = \{\lc a_0\rc\},\
F =\{\dollar\}.$

\item \underline{The edge $\{8,9\}$} has $\beta = w$ and
$ P = \{a^\circ\mid a\in V\},\
F =\{\lc a_0\rc'\}.$

\item \underline{The edge $\{9,10\}$} has $\beta = w$ and
$ P = \{a_{n+1}\},\
F =U\setminus V.$

\item \underline{The edge $\{9,1\}$} has $\beta = w$ and
$ P =V,\
F =\{a_{n+1}\}\cup (U\setminus V).$
\end{itemize}

We show that $\Gamma$ accepts a word $w$ that does not contain $H$ if an only if $T$ eventually halts on $w$.

Let $w=aby, a,b\in V, y\in V^*$ be a word that does not contain $H$ such that $T$ eventually halts on~$w$.
We show how $w$ can be accepted by $\Gamma$.
At the beginning of the computation $w$ is found in node $1$, where the first symbol $a$ can be replaced by $[\phi(a)]$ but the new string cannot pass the filter of any edge, thus remaining in this node during
the next communication step.
In the next step, we can rewrite $b$ as~$b^{\circ}$, getting the new word $[\phi(a)]b^{\circ}y$ which is sent out to node $2$. Here, the symbol $a_0'$ is inserted to its righthand end obtaining $[\phi(a)]b^{\circ}ya_0'$. This word can only enter node $3$. In this node, the string becomes $[\phi(a)]b^{\circ}ya_0'\dollar $, and goes to node $4$.\\
Let $\phi(a)=a_ix$, for some $1\le i\le n+1$ and $x\in X$. In node $4$, $[a_ix]b^{\circ}ya_0'\dollar $ is first converted into  $[a_{i}x]b^{\circ}ya_0''\dollar $, which remains in this node for the next communication step,
and then into  $\lf a_{i-1}x\rf b^{\circ}ya_0''\dollar $.  This string is sent out to node $5$, where it is transformed into $\lc a_{i-1}x\rc b^{\circ}ya_1'\dollar $, via $\lc a_{i-1}x\rc b^{\circ}ya_0''\dollar $. This string goes to node $6$, where $\dollar $ is deleted, and the string becomes $\lc a_{i-1}x\rc b^{\circ}ya_1'$.
If $i>1$, this string first returns to node $3$, resulting in $\lc a_{i-1}x\rc b^{\circ}ya_1'\dollar $, and then goes to node $4$; in this node, the string is transformed into $\lf a_{i-2}x\rf b^{\circ}ya_1''\dollar $, and sent to node $5$. This process is repeated until a string of the form $\lf a_0x\rf b^{\circ}ya_{i-1}''\dollar $  arrives in node $5$. Here the string becomes $\lc a_0x\rc b^{\circ}ya_{i}\dollar $ and goes to node $6$, where the $\dollar $ symbol is deleted. \\
Now, if $x\neq \eps$ (namely, $x=a_jy$ with $j\geq 1$), the string $\lc a_0x\rc b^{\circ}ya_{i} $ goes to node~$2$, where it becomes $\lc a_0x\rc b^{\circ}ya_{i}a_0' $; then, the string enters node $3$, where it is transformed into $\lc a_0x\rc b^{\circ}ya_{i}a_0'\dollar $. Further, the string enters node $4$, where we obtain the string $\lf a_{j-1}y\rf b^{\circ}ya_{i}a_0''\dollar $, and the process described above is resumed.\\
On the other hand, if $x=\eps$, the string equals $\lc a_0\rc b^{\circ}y\phi(a) $, and goes to node $7$, where it becomes $\lc a_0\rc' b^{\circ}ya_{i} $. From this node, the string can only go to node $8$ where the $\lc a_0\rc'$ symbol is deleted. Then, it enters node $9$ where the $b^\circ$ symbol is deleted, and the string becomes $y\phi(a) $. In this moment, the string can go either to node $1$, provided that it doesn't contain $a_{n+1}$ and the whole procedure we described above is resumed, or to node $10$ and the input string is accepted.

We now argue why the above simulation is the only possible derivation in $T$, so that it 
halts on a word $w$ if and only if $w$ is accepted by $\Gamma$.
In many steps, the derivation stated above is the only possible derivation. However, there are a few cases we need to consider more closely.\\
First of all, in node $1$, only one substitution $a\ra a^{\circ}$ can be performed before the string is sent out, but potentially zero or more than one $a\ra [\phi(a)]$ substitutions. If no such substitution is performed, the resulting string enters node $2$ when it remains forever.
If the outgoing word from node $1$
contains more than one symbol $[x]$, $x\in X$, then after several processing steps, either all of them will be transformed into $\lc a_0\rc$, and the string will reach node $7$, or only some of them will be transformed into $\lc a_0\rc$, and the string will be further blocked in node $6$. On the other hand,
if the current string contains more than one $\lc a_0\rc $ symbol, then it enters node $7$, where exactly one of these symbols is transformed into $\lc a_0\rc'$. Then the string goes to node $3$, where it is blocked.\\
Thus, an accepting computation is only possible if exactly one of each of the symbols $a^{\circ}$ and $[x]$ are present when leaving node $1$. However, both symbols could be on any position of the string. Assume that they do not occupy the first two positions in the way described above. The simulation would then go on as described, 
until a string $y_1\lc a_0\rc' y_2b^{\circ}y_3$ or $y_1b^{\circ}y_2\lc a_0\rc' y_3$,
\hbox{$y_1,y_2,y_3\in (V\setminus\{H\})^*$}  is reached. However, in all of these cases the string will be communicated between nodes $7$, $8$ and $9$ only, thus it will not affect the computation.\\
This covers all possible cases, proving that if $w\in L(\Gamma)$, then $T$ will eventually halt on $w$.
\end{proof}

Since $2$-tag systems are universal \cite{rogo,minsky}, the following corollary is immediate:
\begin{cor} There exists a universal ANEPFC with $10$ nodes.
\end{cor}
This result significantly improves the result reported in \cite{mscs} where a universal ANEPFC with $18$ nodes was constructed.

\section{Decreasing the size of ANEPFCs accepting recursively enumerable languages}

Although $2$-tag systems efficiently simulate deterministic Turing machines, via
cyclic tag systems (see, e.\,g., \cite{woods}),
the previous result does not allow us to infer a bound on the size of the networks accepting
in a computationally efficient way all recursively enumerable languages.
We now discuss how an efficient
ANEPFC accepting (deciding) every recursively enumerable (recursive) language can be constructed.

\begin{theorem}\label{TM}
For any recursively enumerable (recursive) language $L$ accepted (decided) by a Turing machine
there exists a complete ANEPFC $\Gamma$ of size $16$ accepting (deciding) $L$.\\
Moreover, if $L\in \mathit{NTIME}(f(n))$, then \hbox{$Time_\Gamma(n)\in {\cal O}(f(n))$}.
\end{theorem}
\begin{proof}
Due to space requirements, we present a sketch of the proof only. However, the principles of this construction rely on the same mechanisms as the ones we used in the proof of Theorem \ref{2tag}, so we believe the reader is able to infer a clear idea of how the construction works.

Let $M=(Q,V,W,q_0,B,F,\delta)$ be a nondeterministic Turing machine; we construct the ANEPFC $\Gamma=(V,U,G,{\cal R},{\mathcal N},\alpha,\beta,1,16)$ with $16$ nodes, labeled with the numbers $1$, $2,\ldots,16$, working as follows.

We stress from the beginning that the edges of the graph, and their filters, are defined such that only the following derivation can take place. Also, nodes $3$, $7$, $8$, $10$ and $14$ do not actually contribute to the simulation of a given derivation, but are used to keep away other derivations from occurring.

In our simulation, a string $v_1[q,a,q',b,X]v_2$ corresponds to a configuration of~$M$ where the current state is $q$, the tape content is  $v_2v_1$ and the head of $M$ reads the first symbol of $v_1$.
To obtain a string having this form corresponding to the initial configuration of $M$
from the input string, two nodes of the network are used: the input node $1$ and node $2$, both of them right insertion nodes. Let now $(q',b,X)\in \delta(q,a)$ be the next transition of $M$.

From a word as above, $\Gamma$ simulates the transition of $M$ in a so-called Simulation Phase. As a first step in this simulation, we need to check that the first symbol of $v_1$ is $a$.  We regard $W$, the working alphabet of Turing machine $M$, as the ordered alphabet $\{a_1,a_2,\ldots,a_n\}$ such that $V=\{a_{1},\ldots,a_m\}$ for some $m<n$ and $B=a_n$. Now we can perform this check by simultaneously lowering the index of the first symbol of $v_1$ and that of $a$ in the symbol $[q,a,q',b,X]$ in the same way as we did in Theorem \ref{2tag}. Only when this check is successful, we obtain a string of the form $\dollar v_1'[q,\dollar,q',b,X]v_2$, where $v_1=av_1'$. The leftmost symbol $\dollar$ is then deleted and we move on to the next stage of the simulation. The procedure described above is carried out by eight of the network's nodes, namely the nodes $3$, $4$, $5$, $6$, $7$, $8$, $9$ and $11$. The node $3$ is a right insertion node, nodes $4$, $5$, $6$ and $8$ are substitution nodes, node $7$ is a deletion node and $9$ is a left deletion node. If the end of the tape (i.\,e., the deleted symbol $a$ equals the blank symbol $B$) is reached during the computation and further tape space is needed, then a special insertion node ensures that the blank symbol $B$ is introduced before continuing the simulation. This is done using node $11$, a left insertion node. Also, node $10$ is used during this step of the computation to collect and block strings that should not be processed further.

Now, $\Gamma$ has to write the symbol $b$ and simulate the correct repositioning of the head for the next move. This is done differently depending on whether the head moves left or right.
If $(q',b,R)\in \delta(q,a)$, it suffices to append the symbol $b$ to the right end of the string. This is done by first inserting a symbol $a_0\notin W$ to the right, to obtain $v_1'[q,\dollar,q',b,R]v_2a_0$. Then, the index of this last symbol is increased, while that of $b$ in the symbol $[q,\dollar,q',b,R]$ is simultaneously lowered. After this process finishes, we have a string of the form $v_1'[q,\dollar,q',\#,R]v_2b$. At this point the simulation is complete, and the symbol $[q,\dollar,q',\#,R]$ can be rewritten as $[q',a',q'',b',X]$ for some $(q'',b',X)\in \delta(q',a')$.

If $(q',b,L)\in \delta(q,a)$, then some more work is necessary. First $b$ is written at the left end of the string, just like explained above. This gives a string of the form $bv_1'[q,\dollar,q',\#,L]v_2$.  Moreover, since now the head reads the last symbol of $v_2$, this symbol has to be moved to the left end of the string. Let $v_2=v_2'c$. Now a symbol $a_0'$ is inserted to the left, giving  $a_0'bv_1'[q,\dollar,q',\#,L]v_2'c$. Again, the index of the first symbol is increased, while that of the last symbol is decreased; moreover, the symbol encoding the move of the Turing machine is updated. Finally, we obtain $c'bv_1'[q',c',q'',d,X]v_2'\bot$, for some $(q'',b',X)\in \delta(q',c')$, from which the symbol $\bot$ is deleted.

To perform the procedure described above five of the network's nodes are used: node $12$, a right insertion node, node $13$, a left insertion node, nodes $3$, $4$, $5$, $6$, $7$ mentioned above, and nodes $14$, a substitution node, and $15$, a right deletion node.

From here a new simulation restarts, unless $q'\in F$, in which case the symbol $[q,\dollar,q',\#,X]$ was replaced by a special symbol $\Delta$ and the word can enter the output node $16$, ensuring that $\Gamma$ accepts the word if and only if $M$ accepts it.

Moreover, if $M$ stops on the input string $w$, in $f(|w|)$ steps, then $\Gamma$ stops on the input string $w$ after $f(|w|)$ executions of the Simulation Phase, described above.
\end{proof}

We can easily state the following corollary of the previous Theorem:
\begin{cor}${\bf NP} = {\bf PTime}_{\ANEPFC_{16}}$.
\end{cor}

This result provides a significantly improvement of the results in \cite{ijfcsDM}, where the same
characterization of {\bf NP} was obtained for ANEPFCs with $26$ nodes.

\bibliographystyle{eptcs}
\bibliography{loos}

\end{document}